\documentclass{mscs}

\title{Modal Dependent Type Theory and Dependent Right Adjoints}
\author[Birkedal, Clouston, Mannaa, M{\o}gelberg, Pitts, and Spitters]
  {L\ls A\ls R\ls S\ns B\ls I\ls R\ls K\ls E\ls D\ls A\ls L$^1$ \ns
   R\ls A\ls N\ls A\ls L\ls D\ns C\ls L\ls O\ls U\ls S\ls T\ls O\ls N$^2$\ns
   B\ls A\ls S\ls S\ls E\ls L\ns M\ls A\ls N\ls N\ls A\ls A$^3$\ns\cr
   R\ls A\ls S\ls M\ls U\ls S\ns E\ls J\ls L\ls E\ls R\ls S\ns M\ls {\O}\ls G\ls E\ls L\ls B\ls E\ls R\ls G$^4$\ns
   A\ls N\ls D\ls R\ls E\ls W\ns M.\ns P\ls I\ls T\ls T\ls S$^5$\ns
   B\ls A\ls S\ns S\ls P\ls I\ls T\ls T\ls E\ls R\ls S$^1$
   \addressbreak
   \addressbreak
   $^1$ Department of Computer Science, Aarhus University, Denmark
\addressbreak
$^2$ Research School of Computer Science, Australian National University, Australia
\addressbreak 
$^3$ Concordium, Denmark
\addressbreak
$^4$ IT University of Copenhagen, Denmark
\addressbreak 
$^5$ Department of Computer Science and Technology, University of Cambridge, UK
}

\date{\today}

\usepackage{harvard}

\usepackage{amsmath}
\usepackage{amssymb} 
\usepackage{ifxetex}
\usepackage[T1]{fontenc}
\usepackage[utf8]{inputenc}
\usepackage[cal=boondox]{mathalfa}
\usepackage{url}
\usepackage{todonotes}

\newtheorem{theorem}{Theorem}
\newtheorem{lemma}[theorem]{Lemma}
\newtheorem{proposition}[theorem]{Proposition}
\newtheorem{definition}[theorem]{Definition}
\newtheorem{corollary}[theorem]{Corollary}
\newtheorem{remark}[theorem]{Remark}

\usepackage{mathpartir}

\usepackage{fontawesome}
\usepackage[mathscr]{euscript}
\usepackage{tikz-cd}

\newcommand{\cat}[1]{\mathbf{#1}}

\newcommand{\id}{\mathsf{id}}
\newcommand{\op}[1]{#1^{\mathsf{op}}}

\newcommand{\inv}[1]{#1^{-1}}

\newcommand{\adj}{\dashv}

\newcommand{\Emp}{\diamond}
\newcommand{\ent}{\vdash}
\newcommand{\fun}{\rightarrow}
\newcommand{\comp}{\circ}
\newcommand{\T}{\langle\rangle}

\newcommand{\U}{\mathsf{U}}
\newcommand{\code}[1]{\ulcorner#1\urcorner}
\DeclareMathOperator{\E}{\mathsf{E}}
\newcommand{\pr}{\mathsf{pr}}
\renewcommand{\r}{\mathsf{r}}
\renewcommand{\l}{\ell}
\DeclareMathOperator{\R}{\mathsf{R}}
\DeclareMathOperator{\Ran}{\mathrm{Ran}}
\newcommand{\Rl}{\mathsf{Rl}}
\usepackage{stmaryrd}
\usepackage[all,pdf]{xy}
\CompileMatrices 
\newcommand{\pullbackcorner}[1][dr]{\save*!/#1-1.4pc/#1:(-1,1)@^{|-}\restore}

\renewcommand{\L}{\mathsf{L}}

\newcommand{\p}[0]{\proj}

\newcommand{\q}[0]{\gen}

\newcommand{\C}{\cat{C}}
\newcommand{\D}{\cat{D}}
\newcommand{\co}[2]{#1{.}#2}
\newcommand{\CWF}{\C}
\newcommand{\Elt}[3][\CWF]{#1(#2\mathbin{\scriptstyle\vdash}#3)} 
\newcommand{\Fam}[2][\CWF]{#1(#2)}
\newcommand{\FamU}[3][\C]{#1(#2,{#3})}
\newcommand{\gen}{\mathsf{q}}
\newcommand{\Hom}[3][\CWF]{#1(#2,#3)}
\newcommand{\morphism}{\rightarrow}
\newcommand{\proj}{\mathsf{p}}

\newcommand{\app}[2]{{app}\left(#1,#2\right)}

\newcommand{\pair}[2]{\left(#1,#2\right)}

\newcommand{\El}[1]{\mathsf{El}\,#1}

\newcommand{\bij}{\cong}

\newcommand{\earlier}{\triangleleft}
\newcommand{\later}{\triangleright}

\newcommand{\fix}{\mathsf{fix}}
\newcommand{\nxt}{\mathsf{next}}

\newcommand{\nats}{\mathbb{N}}

\newcommand{\transp}[1]{\overline{#1}}

\newcommand{\define}[1]{\textbf{#1}}
\usepackage[status=draft]{fixme} 

\renewcommand{\Box}{\square}
\newcommand{\defeq}{\triangleq}
\newcommand{\den}[1]{\llbracket#1\rrbracket}
\newcommand{\eqterm}[4]{#1\vdash#2=#3:#4}
\newcommand{\eqtype}[3]{#1\vdash#2=#3}
\newcommand{\eqtypeU}[4]{#1\vdash_{#4}#2=#3}
\newcommand{\Exch}[4]{\mathsf{E}(#1;#2;#3;#4)}
\newcommand{\isctx}[1]{#1\vdash}
\newcommand{\isterm}[3]{#1\vdash#2:#3}
\newcommand{\istype}[2]{#1\vdash#2}
\newcommand{\istypeU}[3]{#1\vdash_{#3}#2}
\newcommand{\Jud}{\mathcal{J}}
\newcommand{\lock}{\mbox{\faUnlock}}

\newcommand{\open}{\ensuremath{\mathsf{open}}}
\newcommand{\Proj}[3]{\mathsf{P}(#1;#2;#3)}
\newcommand{\PiT}[3]{\Pi #1:#2.\,#3}
\newcommand{\shut}{\ensuremath{\mathsf{shut}}}
\newcommand{\Subst}[4]{\mathsf{S}(#1;#2;#3;#4)}
\newcommand{\subst}[3]{#1[#2/#3]}

\newcommand{\Terminal}{\top}

\newcommand{\Gir}{\mathscr{G}}

\newcommand{\Atom}{\mathbb{A}}
\newcommand{\Nom}{\mathbf{Nom}}

\newcommand{\clott}{\textsf{CloTT}}
\newcommand\latbind[2]{{\triangleright}\, (#1:#2) .}
\newcommand{\tickA}{\alpha}

\newcommand{\CwDRA}{CwDRA}
\newcommand{\CwDRAs}{CwDRAs}
\newcommand{\CwU}{CwU}
\newcommand{\CwUDRA}{CwUDRA}

\begin{document}

\maketitle

\begin{abstract}
  In recent years we have seen several new models of dependent type theory
  extended with some form of modal necessity operator, including nominal type
  theory, guarded and clocked type theory, and spatial
  and cohesive type theory. 
  In this paper we study \emph{modal dependent type theory}:
  dependent type theory with an operator satisfying (a dependent version of)
  the K axiom of modal logic.
  We investigate both semantics and syntax.
  For the semantics, we introduce categories with families with a
  \emph{dependent right adjoint} (\CwDRA) and show that the examples
  above can be presented as such.
  Indeed, we show that any category with finite limits and an adjunction of
  endofunctors gives rise to a \CwDRA\ via the local universe construction.
  For the syntax, we introduce a dependently typed extension of Fitch-style
  modal lambda-calculus, show that it can be interpreted in any \CwDRA, and
  build a term model.
  We extend the syntax and semantics with universes.
\end{abstract}

\section{Introduction}
\label{sec:intro}

Dependent types are a powerful technology for both programming and formal
proof. In recent years we have seen several new models of dependent type theory extended with a type former resembling modal necessity%
\footnote{For an introduction to modal logic, see e.g.~\citeasnoun{blackburn2002modal}.}%
, such as nominal
type theory~\cite{PittsAM:deptta},
guarded~\cite{birkedal2011first,BirkedalL:gdtt-conf,GDTTmodel,GCTT} and
clocked~\cite{CloTTmodel} type theory, and spatial and cohesive type
theory~\cite{shulman2018brouwer}. These examples all satisfy the K axiom of
modal logic
\begin{equation}\label{axiom_K}
\Box(A\to B)\to\Box A\to\Box B
\end{equation}
but are not all (co)monads, the more extensively studied construction in the
context of dependent type
theory~\cite{Krishnaswami:Integrating,dePaiva:Fibrational,vakar2017search,shulman2018brouwer}.
Motivated in part by these examples, in this paper we study \emph{modal dependent type theory}: dependent type theory with an operator satisfying (a dependent generalisation of) the K axiom%
\footnote{For ``Kripke''; not to be confused with Streicher's K~\cite{streicher1993investigations}.}
of modal logic. We investigate both semantics and syntax.

For the semantics, we introduce categories with families with a \emph{dependent right adjoint} (\CwDRA) and show that this dependent right
adjoint models the modality in the examples mentioned above.
Indeed, we show that any finite limit category with an adjunction of endofunctors%
\footnote{This should not be confused with models where there are
adjoint functors between different categories which can be composed to
define a monad or comonad.}
gives rise to a \CwDRA\ via the local universe
construction~\cite{LumsdainePL:locumo}.
In particular, by applying the local universe construction to a locally cartesian closed category with an adjunction of endofunctors, we get a model of modal dependent type theory with $\Pi$- and $\Sigma$-types. 

For the syntax, we adapt the simply typed Fitch-style modal lambda-calculus
introduced by~\citeasnoun{Borghuis:Coming} and~\citeasnoun{Martini:Computational}, inspired by Fitch's proof theory for modal
logic~\cite{Fitch:Symbolic}. In such a calculus $\Box$ is introduced by `shutting'
a strict subordinate proof and eliminating by `opening' one. For example the
K axiom~\eqref{axiom_K} is inhabited by the term
\begin{equation}\label{eq:K_inhabited}
\lambda f.\lambda x.\shut((\open\,f)(\open\,x))
\end{equation}
The nesting of subordinate proofs can be tracked in sequent style by a special
symbol in the context which we call a \emph{lock}, and write $\lock$; the open
lock symbol is intended to suggest we have access to the contents of a box.
Following~\citeasnoun{Clouston:fitch-2018}, the lock can be understood
as an operation on contexts \emph{left adjoint} to $\Box$; hence Fitch-style
modal $\lambda$-calculus has a model in any cartesian closed category equipped
with an adjunction of endofunctors. Here we show, in work inspired by Clocked
Type Theory~\cite{bahr2017clocks}, that Fitch-style $\lambda$-calculus lifts
with a minimum of difficulty to dependent types. In particular the term
\eqref{eq:K_inhabited}, where $f$ is a dependent function, has type
\[
  \Box(\Pi y:A.\,B) \to \Pi x:\Box A.\,\Box\subst{B}{\open\,x}{y}
\]
This dependent version of the K axiom, not obviously expressible without the
$\open$ construct of a Fitch-style calculus, allows modalised functions to be
applied to modalised data even in the dependent case. This capability is known
to be essential in at least one example, namely proofs about guarded
recursion~\cite{BirkedalL:gdtt-conf}%
\footnote{This capability was achieved by Bizjak et
al.~\cite{BirkedalL:gdtt-conf} via \emph{delayed substitutions}, but this
construction does not straightforwardly support an operational
semantics~\cite{bahr2017clocks}.}%
.
We show that our calculus can be soundly interpreted in any \CwDRA, and
construct a term model.

We also extend the syntax and semantics of modal dependent type theory with universes.
Here we restrict attention to models based on (pre)sheaves, for which Coquand has proposed a particularly simple formulation of universes~\cite{Coquand:CwU}.
We show how to extend Coquand's notion of a category with universes with dependent right adjoints, and observe that a construction encoding the modality on the universe,
introduced for guarded type theory by~\citeasnoun{BirkedalL:gdtt-conf}, in fact arises for
more general reasons.

Another motivation for the present work is that it can be understood as providing a notion of a \emph{dependent} adjunction between endofunctors.
An ordinary adjunction $\L\dashv \R$ on a category $\C$ is a natural bijective
correspondence $\Hom{\L A}{B} \bij \Hom{A}{\R B}$.
With dependent types one might consider \emph{dependent} functions from
$\L A$ to $B$, where $B$ may depend on $\L A$, and similarly from $A$ to
$\R B$. Our notion of \CwDRA\ then defines what it means to have an adjoint correspondence in this dependent case. Our Fitch-style modal dependent type
theory can therefore also be understood as a term language for dependent
adjoints.

\medskip\noindent\textbf{Outline}\quad We introduce \CwDRAs\ in Section~\ref{sec:CwDRA}, and present the syntax of modal dependent type theory in Section~\ref{sec:syntax}.
In Section~\ref{sec:contruction-of-CwDRAs} we show how to construct a \CwDRA\ from an adjunction on a category with finite limits.
In Section~\ref{sec:examples} we show how various models in the literature
can be presented as \CwDRAs.
The extension with universes is defined in Section~\ref{sec:universes}.
We end with a discussion of related and future work in Section~\ref{sec:discussion}.

\section{Categorical Semantics of Modal Dependent Type Theory}
\label{sec:CwDRA}

The notion of \emph{category with families}
(CwF)~\cite{Dybjer1996,hofmann1997syntax} provides a semantics for the
development of dependent type theory which elides some difficult aspects of
syntax, such as variable binding, as well as the coherence problems of simpler
notions of model.  It can be connected to syntax by a soundness argument and
term model construction, and to more mathematical models via `strictification'
constructions~\cite{hofmann1994interpretation,LumsdainePL:locumo}. 
In this section we extend this notion to introduce
\emph{categories with a dependent right adjoint} (CwDRA). We first recall the
standard definition:

\begin{definition}[\textbf{category with families}]
  \label{def:cwf}
  A \define{CwF} is specified by:
  \begin{enumerate}
  \item A category $\C$ with a terminal object $\Terminal$. Given objects
    $\Gamma,\Delta\in\C$, write $\C(\Delta,\Gamma)$ for the set of
    morphisms from $\Delta$ to $\Gamma$ in $\C$. The identity morphism
    on $\Gamma$ is just written $\id$ with $\Gamma$ implicit. The
    composition of $\gamma\in\C(\Delta,\Gamma)$ with
    $\delta\in\C(\Phi,\Delta)$ is written $\gamma\comp\delta$.
    
  \item For each object $\Gamma\in\C$, a set $\Fam{\Gamma}$ of
    \emph{families} over $\Gamma$.

  \item For each object $\Gamma\in\C$ and family $A\in\Fam{\Gamma}$, a
    set $\Elt{\Gamma}{A}$ of \emph{elements} of the family $A$ over
    $\Gamma$.

  \item\label{item:2} For each morphism
    $\gamma\in\Hom{\Delta}{\Gamma}$, \emph{re-indexing}
    functions $A\in\Fam{\Gamma} \mapsto A[\gamma]\in\Fam{\Delta}$ and
    $a\in\Elt{\Gamma}{A}\mapsto a[\gamma]\in\Elt{\Delta}{A[\gamma]}$,
    satisfying $A[\id]=A$, $A[\gamma\comp\delta] = A[\gamma][\delta]$,
    $a[\id]=a$ and $a[\gamma\comp\delta] = a[\gamma][\delta]$.

  \item\label{item:1} For each object $\Gamma\in\C$ and family
    $A\in\Fam{\Gamma}$, a \emph{comprehension object}
    $\co{\Gamma}{A}\in\C$ equipped with a \emph{projection morphism}
    $\proj_A\in\C(\co{\Gamma}{A},\Gamma)$, a \emph{generic element}
    $\gen_A\in\Elt{\co{\Gamma}{A}}{A[\proj_A]}$ and a \emph{pairing
      operation} mapping 
    $\gamma\in\Hom{\Delta}{\Gamma}$ and $a\in\Elt{\Delta}{A[\gamma]}$
    to $\pair{\gamma}{a}\in\C(\Delta,\co{\Gamma}{A})$ satisfying
    $\proj_A\comp\pair{\gamma}{a}= \gamma$,
    $\gen_A[\pair{\gamma}{a}] = a$,
    $\pair{\gamma}{a}\comp\delta =
    \pair{\gamma\comp\delta}{a[\delta]}$ and
    $\pair{\proj_A}{\gen_A} = \id$.
  \end{enumerate}
\end{definition}

A \emph{dependent right adjoint} then extends the definition of CwF with a
functor on contexts $\L$ and an operation on families $\R$, intuitively
understood to be left and right adjoints:

\begin{definition}[\textbf{category with a dependent right adjoint}]
  \label{def:cwf-a}
  A \emph{CwDRA} is a CwF $\C$ equipped with the following extra
  structure:

  \begin{enumerate}
  \item\label{item:3} An endofunctor $\L:\C\morphism\C$ on the
    underlying category of the CwF.
    
  \item For each object $\Gamma\in\C$ and family $A\in\Fam{\L\Gamma}$,
    a family $\R_\Gamma A \in\Fam{\Gamma}$, stable under re-indexing
    in the sense that for all $\gamma\in\C(\Delta,\Gamma)$ we have
    \begin{equation}
      \label{eq:11}
      (\R_\Gamma A)[\gamma] = \R_\Delta(A[\L\gamma]) \in \Fam{\Delta}
    \end{equation}
    
  \item For each object $\Gamma\in\C$ and family $A\in\Fam{\L\Gamma}$
    a bijection
    \begin{equation}
      \label{eq:18}
      \Elt{\L\Gamma}{A}\bij\Elt{\Gamma}{\R_\Gamma A}
    \end{equation}
    We write the effect of this bijection on $a\in\Elt{\L\Gamma}{A}$
    as $\transp{a}\in\Elt{\Gamma}{\R_\Gamma A}$ and write the effect of
    its inverse on $b\in\Elt{\Gamma}{\R_\Gamma A}$ also as
    $\transp{b}\in\Elt{\L\Gamma}{A}$. Thus
    \begin{align}
      \transp{\transp{a}} &= a &&(a\in\Elt{\L\Gamma}{A}) \label{eq:12}\\
      \transp{\transp{b}} &= b &&(b\in\Elt{\Gamma}{\R_\Gamma A})\label{eq:13}     
    \end{align}
    The bijection is required to be stable under
    re-indexing in the sense that for all $\gamma\in\C(\Delta,\Gamma)$
    we have
    \begin{equation}
      \label{eq:14}
      \transp{a}[\gamma] = \transp{a[\L\gamma]}
    \end{equation}
  \end{enumerate}
\end{definition}
Note that equation (\ref{eq:14}) is well-typed by (\ref{eq:11}). Equation (\ref{eq:14})
also implies that the opposite direction of the isomorphism (\ref{eq:18}) is natural, i.e., that
the equation 
    \begin{equation}
      \label{eq:15}
      \transp{b}[\L\gamma] = \transp{b[\gamma]}
    \end{equation}
also holds, since 
      $\transp{b[\gamma]} = \transp{\transp{\transp{b}}[\gamma]} =
      \transp{\transp{\transp{b}[\L\gamma]}} = \transp{b}[\L\gamma]$.

\section{Syntax of Modal Dependent Type Theory}
\label{sec:syntax}

In this section we extend Fitch-style modal
$\lambda$-calculus~\cite{Borghuis:Coming} to dependent types, and connect this to the
notion of CwDRA via a soundness proof and term model construction.
We define our dependent types broadly in the style of ECC~\cite{Luo:ECC}, as this is close to the implementation of some proof assistants~\cite{norell:thesis}.

We define the raw syntax of contexts, types, and terms as follows:
\begin{align*}
  \Gamma &\;\defeq\; \Emp \;\mid\; \Gamma,x:A \;\mid\; \Gamma,\lock \\
  A &\;\defeq\; \PiT{x}{A}{B} \;\mid\; \Box A \\
  t &\;\defeq\; x \;\mid\; \lambda x:A.\,t \;\mid\; t\,t \;\mid\; \shut\,t \;\mid\; \open\,t
\end{align*}
We omit the leftmost `$\Emp,$' where the context is non-empty. We will usually omit the
type annotation on the $\lambda$ for brevity. $\Pi$-types
are included in the grammar as an example to show that standard
type formers can be defined as usual, without reference to the locks
in the context. One could similarly add an empty type, unit type, booleans,
$\Sigma$-types, W-types, universes (of which more in
Section~\ref{sec:universes}), and so forth.

Judgements have forms
\begin{align*}
  \isctx{\Gamma} &\qquad\mbox{`$\Gamma$ is a well-formed context'} \\
  \istype{\Gamma}{A}    &\qquad\mbox{`$A$ is a well-formed type in context $\Gamma$'} \\
  \eqtype{\Gamma}{A}{B} &\qquad\mbox{`$A$ and $B$ are equal types in context $\Gamma$'} \\
  \isterm{\Gamma}{t}{A} &\qquad\mbox{`$t$ is a term with type $A$ in context $\Gamma$'} \\
  \eqterm{\Gamma}{t}{u}{A} &\qquad\mbox{`$t$ and $u$ are equal terms with type $A$ in context $\Gamma$'}
\end{align*}

\begin{figure}
Context formation rules:
\begin{mathpar}
  \inferrule*{ }{\isctx{\Emp}}
  \and
  \inferrule*[right=$x\notin\Gamma$]{
    \isctx{\Gamma} \\
    \istype{\Gamma}{A}}{
    \isctx{\Gamma,x:A}}
  \and
  \inferrule*{
    \isctx{\Gamma}}{
    \isctx{\Gamma,\lock}}
  \and
  \inferrule*[right=$x$ not free in $B$]{
    \isctx{\Gamma,x:A,y:B,\Gamma'}}{
    \isctx{\Gamma,y:B,x:A,\Gamma'}}
\end{mathpar}

Type formation rules:
\begin{mathpar}
  \inferrule*{
    \istype{\Gamma}{A} \\
    \istype{\Gamma,x:A}{B}}{
    \istype{\Gamma}{\PiT{x}{A}{B}}}
  \and
  \inferrule*{
    \istype{\Gamma,\lock}{A}}{
    \istype{\Gamma}{\Box A}}
\end{mathpar}

Type equality rules are as standard, asserting equivalence, and congruence with respect
to all type formers.

Term formation rules:
\begin{mathpar}
  \inferrule*{
    \isterm{\Gamma}{t}{A} \\
    \eqtype{\Gamma}{A}{B}}{
    \isterm{\Gamma}{t}{B}}
  \and
  \inferrule*[right=$\lock\notin\Gamma'$]{
    \isctx{\Gamma,x:A,\Gamma'}}{
    \isterm{\Gamma,x:A,\Gamma'}{x}{A}}
  \and
  \inferrule*{
    \isterm{\Gamma,x:A}{t}{B}}{
    \isterm{\Gamma}{\lambda x.t}{\PiT{x}{A}{B}}}
  \and
  \inferrule*{
    \isterm{\Gamma}{t}{\PiT{x}{A}{B}} \\
    \isterm{\Gamma}{u}{A}}{
    \isterm{\Gamma}{t\,u}{\subst{B}{u}{x}}} \\
  \and
  \inferrule*{
    \isterm{\Gamma,\lock}{t}{A}}{
    \isterm{\Gamma}{\shut\,t}{\Box A}}
  \and
  \inferrule*[right=$\lock\notin\Gamma'$]{
    \isterm{\Gamma}{t}{\Box A} \\
    \isctx{\Gamma,\lock,\Gamma'}}{
    \isterm{\Gamma,\lock,\Gamma'}{\open\,t}{A}}
\end{mathpar}

Term equality rules, omitting equivalence and congruence:
\begin{mathpar}
  \inferrule*{
    \isterm{\Gamma}{(\lambda x.t)u}{A}}{
    \eqterm{\Gamma}{(\lambda x.t)u}{\subst{t}{u}{x}}{A}}
  \and
  \inferrule*{
    \isterm{\Gamma}{\open\,\shut\,t}{A}}{
    \eqterm{\Gamma}{\open\,\shut\,t}{t}{A}}
  \and
  \inferrule*[right=$x\notin\Gamma$]{
    \isterm{\Gamma}{t}{\PiT{x}{A}{B}}}{
    \eqterm{\Gamma}{t}{\lambda x.t\,x}{\PiT{x}{A}{B}}}
  \and
  \inferrule*{
    \isterm{\Gamma}{t}{\Box A}}{
    \eqterm{\Gamma}{t}{\shut\,\open\,t}{\Box A}}
\end{mathpar}
\caption{Typing rules for a dependent Fitch-style modal $\lambda$-calculus.}
\label{fig:typing_rules}
\end{figure}

Figure~\ref{fig:typing_rules} presents the typing rules of the calculus.
The syntactic results below follow easily by induction on these rules. We
remark only that exchange of variables with locks, and weakening of locks, are
not admissible, and that the (lock-free) weakening $\Gamma'$ in the $\open$
rule is essential to proving variable weakening.

\begin{lemma}\label{lem:syntactic_sanity}
Let $\Jud$ range over the possible strings to the right of a turnstile in a
judgement.
\begin{enumerate}
\item If $\Gamma,x:A,y:B,\Gamma'\vdash\Jud$ and $x$ is not free in $B$, then
$\Gamma,y:B,x:A,\Gamma'\vdash\Jud$;
\item If $\Gamma,\Gamma'\vdash\Jud$, and $\Gamma\vdash A$, and $x$ is a
fresh variable, then $\Gamma,x:A,\Gamma'\vdash\Jud$;
\item If $\Gamma,x:A,\Gamma'\vdash \Jud$ and $\Gamma\vdash u:A$, then
$\Gamma,\subst{\Gamma'}{u}{x}\vdash\subst{\Jud}{u}{x}$;
\item If $\Gamma\vdash t:A$ then $\Gamma\vdash A$;
\item If $\eqterm{\Gamma}{t}{u}{A}$ then $\isterm{\Gamma}{t}{A}$ and
$\isterm{\Gamma}{u}{A}$.
\end{enumerate}
\end{lemma}

\subsection{Sound interpretation in \CwDRAs}\label{sec:interpret}

In this section we show that the calculus of Figure~\ref{fig:typing_rules} can be soundly
interpreted in any CwDRA. We wish to give meaning to contexts, types, and terms, but
(via the type conversion rule) these can have multiple derivations, so it is not possible to
work by induction on the formation rules. Instead, following
e.g.~\citeasnoun{hofmann1997syntax}, we define a partial map from \emph{raw
syntax} to semantics by induction on the grammar, then prove this map is
defined for well-formed syntax. By `raw syntax' we mean contexts, types accompanied by a context, and terms accompanied by context
and type, defined via the grammar. The \emph{size} of a type or term is the
number of connectives and variables used to define it, and the size of a context is the sum of
the sizes of its types.

Well-defined contexts $\Gamma$ will be interpreted as objects $\den{\Gamma}$ in $\CWF$,
types in context $\istype{\Gamma}{A}$ as families in $\Fam{\den{\Gamma}}$, and typed
terms in context $\isterm{\Gamma}{t}{A}$ as elements in
$\Elt{\den{\Gamma}}{\den{\istype{\Gamma}{A}}}$. Where there is no
confusion we write $\den{\istype{\Gamma}{A}}$ as $\den{A}$ and
$\den{\isterm{\Gamma}{t}{A}}$ as $\den{\Gamma\vdash t}$ or
$\den{t}$.

The partial interpretation of raw syntax is as follows, following the convention
that ill-formed expressions (for example, where a subexpression is undefined)
are undefined. We omit the details for $\Pi$-types and other standard
constructions, which are as usual.
\begin{itemize}
\item $\den{\Emp}=\Terminal$;
\item $\den{\Gamma,x:A}=\co{\den{\Gamma}}{\den{A}}$;
\item $\den{\Gamma,\lock}=\L\den{\Gamma}$;
\item $\den{\Gamma\vdash\Box A}=R_{\den{\Gamma}}(\den{A})$;
\item $\den{\isterm{\Gamma,x:A,x_1:A_1,\ldots,x_n:A_n}{x}{A}}=
\gen_{\den{A}}[\proj_{\den{A_1}}\circ\cdots\circ\proj_{\den{A_n}}]$;
\item $\den{\isterm{\Gamma}{\shut\,t}{\Box A}}=\transp{\den{t}}$;
\item $\den{\isterm{\Gamma,\lock,x_1:A_1,\ldots,x_n:A_n}{\open\,t}{A}}=
\transp{\den{t}}[\proj_{\den{A_1}}\circ\cdots\circ\proj_{\den{A_n}}]$.
\end{itemize}

In Figure~\ref{fig:struct_morphisms} we define expressions $\Proj{\Gamma}{A}{\Gamma'}$,
$\Exch{\Gamma}{A}{B}{\Gamma'}$, and $\Subst{\Gamma}{A}{\Gamma'}{t}$ that, where
defined, define morphisms in $\C$ corresponding respectively to weakening, exchange, and
substitution in contexts.

\begin{figure}
  \begin{tabular}{l c l}
    $\Proj{\Gamma}{A}{\Emp}$ &$=$& $\proj_{\den{A}}$ \\
    $\Proj{\Gamma}{A}{\Gamma',y:B}$ &$=$&
      $\pair{\Proj{\Gamma}{A}{\Gamma'}\circ
      \proj_{\den{\istype{\Gamma,x:A,\Gamma'}{B}}}}
      {\gen_{\den{\istype{\Gamma,x:A,\Gamma'}{B}}}}$ \\
    $\Proj{\Gamma}{A}{\Gamma',\lock}$ &$=$&
      $\L\,\Proj{\Gamma}{A}{\Gamma'}$ \vspace {0.25em} \\
    $\Exch{\Gamma}{A}{B}{\Emp}$ &$=$&
      $((\proj_{\den{\istype{\Gamma}{B}}}\circ
      \proj_{\den{\istype{\Gamma,y:B}{A}}},
      \gen_{\den{\istype{\Gamma,y:B}{A}}}),\gen_{\den{\istype{\Gamma}{B}}}
      [\proj_{\den{\istype{\Gamma,y:B}{A}}}])$ \\
    $\Exch{\Gamma}{A}{B}{\Gamma',z:C}$ &$=$&
      $(\Exch{\Gamma}{A}{B}{\Gamma'}\circ
      \proj_{\den{\istype{\Gamma,y:B,x:A,\Gamma'}{C}}},
      \gen_{\den{\istype{\Gamma,y:B,x:A,\Gamma'}{C}}})$ \\
    $\Exch{\Gamma}{A}{B}{\Gamma',\lock}$ &$=$&
      $\L\,\Exch{\Gamma}{A}{B}{\Gamma'}$ \vspace {0.25em} \\
    $\Subst{\Gamma}{A}{\Emp}{t}$ &$=$& $(\id,\den{t})$ \\
    $\Subst{\Gamma}{A}{\Gamma',y:B}{t}$ &$=$&
      $(\Subst{\Gamma}{A}{\Gamma'}{t}\circ
      \proj_{\den{\istype{\Gamma,\subst{\Gamma'}{t}{x}}{\subst{B}{t}{x}}}},
      \gen_{\den{\istype{\Gamma,\subst{\Gamma'}{t}{x}}{\subst{B}{t}{x}}}})$ \\
    $\Subst{\Gamma}{A}{\Gamma',\lock}{t}$ &$=$&
      $\L\,\Subst{\Gamma}{A}{\Gamma'}{t}$
  \end{tabular}
  \caption{Morphisms in $\C$ corresponding to weakening, exchange, and substitution.}
\label{fig:struct_morphisms}
\end{figure}

\begin{lemma}\label{lem:Proj}
Suppose $\den{\Gamma,\Gamma'}$ and $\den{\Gamma,x:A,\Gamma'}$ are
defined. Then the following properties hold:
\begin{enumerate}
\item $\den{\Gamma,x:A,\Gamma'\vdash X}\simeq
\den{\Gamma,\Gamma'\vdash X}[\Proj{\Gamma}{A}{\Gamma'}]$, where
$\simeq$ is Kleene equality, and $X$ is a type or typed term;
\item $\Proj{\Gamma}{A}{\Gamma'}$ is a well-defined morphism from
$\den{\Gamma,x:A,\Gamma'}$ to $\den{\Gamma,\Gamma'}$;
\end{enumerate}
\end{lemma}
\begin{proof}
The proof proceeds by mutual induction on the size of $\Gamma'$ (for
statement 2) and the size of $\Gamma'$ plus the size of $X$ (for statement 1).
We present only the cases particular to $\Box$.

We start with statement 1. We use the mutual induction with statement
2 at the smaller size of $\Gamma'$ alone to ensure that
$\Proj{\Gamma}{A}{\Gamma'}$ is well-formed with the correct domain and
codomain, then proceed by induction on the construction of $X$.

The $\Box$ case follows because
\begin{align*}
  \den{\Gamma,x:A,\Gamma'\vdash\Box B} &\simeq
    \R_{\den{\Gamma,x:A,\Gamma'}}\den{\Gamma,x:A,\Gamma',\lock\vdash B} \\
  &\simeq \R_{\den{\Gamma,x:A,\Gamma'}}(\den{\Gamma,\Gamma',\lock\vdash B}
[\Proj{\Gamma}{A}{\Gamma',\lock}]) \tag{induction} \\
  &\simeq R_{\den{\Gamma,x:A,\Gamma'}}(\den{\Gamma,\Gamma',\lock\vdash B}
[\L\Proj{\Gamma}{A}{\Gamma'}]) \\
  &\simeq (\R_{\den{\Gamma,\Gamma'}}\den{\Gamma,\Gamma',\lock\vdash B})
[\Proj{\Gamma}{A}{\Gamma'}] \tag{\ref{eq:11}} \\
  &= \den{\Gamma,\Gamma'\vdash\Box B}[\Proj{\Gamma}{A}{\Gamma'}]
\end{align*}

The $\shut$ case follows immediately from \eqref{eq:14} and induction.
For $\open$, the case where the deleted variable $x$ is to the right of the lock
follows by Definition~\ref{def:cwf} part 5. Suppose instead it is to the left.
Then
\begin{align*}
  &\den{\Gamma,\Gamma',\lock,y_1:B_1,\ldots,y_n:B_n\vdash\open\,t}
[\Proj{\Gamma}{A}{\Gamma',\lock,y_1:B_1,\ldots,y_n:B_n}] \\
  &\simeq \transp{\den{t}}[\proj_{\den{B_1}}\circ\cdots\circ\proj_{\den{B_n}}\circ
\Proj{\Gamma}{A}{\Gamma',\lock,y_1:B_1,\ldots,y_n:B_n}] \\
  &\simeq \transp{\den{t}}[\Proj{\Gamma}{A}{\Gamma',\lock}\circ
\proj_{\den{B_1}}\circ\cdots\circ\proj_{\den{B_n}}] \tag{Definition~\ref{def:cwf} part 5} \\
  &\simeq \transp{\den{t}}[\L\Proj{\Gamma}{A}{\Gamma'}]
[\proj_{\den{B_1}}\circ\cdots\circ\proj_{\den{B_n}}] \\
  &\simeq \transp{\den{t}[\Proj{\Gamma}{A}{\Gamma'}]}
[\proj_{\den{B_1}}\circ\cdots\circ\proj_{\den{B_n}}] \tag{\ref{eq:15}} \\
  &\simeq \transp{\den{\Gamma,x:A,\Gamma'\vdash t}}
[\proj_{\den{B_1}}\circ\cdots\circ\proj_{\den{B_n}}] \tag{induction} \\
  &\simeq {\den{\Gamma,x:A,\Gamma',\lock,y_1:B_1,\ldots,y_n:B_n\vdash \open\,t}}
\end{align*}

For statement 2, the lock case holds immediately by application of the functor
$\L$. Other cases follow as standard; for example the base case holds
because $\proj_{\den{A}}$ is indeed a morphism. 
\end{proof}

\begin{lemma}\label{lem:exchange}
Suppose $\den{\Gamma,x:A,y:B,\Gamma'}$ and $\den{\Gamma\vdash B}$ are
defined. Then the following properties hold:
\begin{enumerate}
\item $\den{\istype{\Gamma,y:B,x:A,\Gamma'}{X}}\simeq
\den{\istype{\Gamma,x:A,y:B,\Gamma'}{X}}
[\Exch{\Gamma}{A}{B}{\Gamma'}]$, where $X$ is a type or typed term;
\item $\Exch{\Gamma}{A}{B}{\Gamma'}$ is a well-defined morphism from
$\den{\Gamma,y:B,x:A,\Gamma'}$ to $\den{\Gamma,x:A,y:B,\Gamma'}$;
\end{enumerate}
\end{lemma}
\begin{proof}
The base case of statement 1 uses Lemma~\ref{lem:Proj}; the proof otherwise
follows just as with Lemma~\ref{lem:Proj}.
\end{proof}

\begin{lemma}\label{lem:Subst}
Suppose $\den{\Gamma\vdash t:A}$ and
$\den{\Gamma,x:A,\Gamma'}$ are defined. Then the following properties hold:
\begin{enumerate}
\item $\den{\Gamma,\subst{\Gamma'}{t}{x}\vdash\subst{X}{t}{x}}\simeq
\den{\Gamma,x:A,\Gamma'\vdash X}[\Subst{\Gamma}{A}{\Gamma'}{t}]$,
where $X$ is a type or typed term;
\item $\Subst{\Gamma}{A}{\Gamma'}{t}$ is a well-defined morphism from
$\den{\Gamma,\subst{\Gamma'}{t}{x}}$ to $\den{\Gamma,x:A,\Gamma'}$;
\end{enumerate}
\end{lemma}
\begin{proof}
As with Lemma~\ref{lem:Proj}.
\end{proof}

\begin{theorem}[Soundness]
Where a context, type, or term is well-formed, its denotation is well-defined,
and all types and terms identified by equations have the same denotation. 
\end{theorem}
\begin{proof}
Most cases follow as usual, using Lemmas~\ref{lem:Proj},
\ref{lem:exchange}, and~\ref{lem:Subst} as needed. The well-definedness of the
formation rules for $\Box$ are straightforward, so we present only the equations
for $\Box$:

Starting with $\isterm{\Gamma,\lock}{t}{A}$ we have $\isterm{\Gamma,
\lock,x_1:A_1,\ldots,x_n:A_n}{\open\,\shut\,t}{A}$ and wish to prove its
denotation is equal to that of $t$ (with the weakening $x_1,\ldots,x_n$). Then
$\den{\open\,\shut\,t}=\transp{\transp{\den{t}}}[\proj_{\den{A_1}}\circ
\cdots\circ\proj_{\den{A_n}}]=\den{t}[\proj_{\den{A_1}}\circ\cdots
\circ\proj_{\den{A_n}}]$, which is the weakening of $t$ by
Lemma~\ref{lem:Proj}.

The equality of $\den{\shut\,\open\,t}$ and $\den{t}$ is straightforward.
\end{proof}

\subsection{Term model}\label{sec:term_model}

We now develop as our first example of a CwDRA, a term model built from the
syntax of our calculus. The objects of this category are contexts modulo
equality, which is defined pointwise via type equality. We define an arrow
$\Delta\to\Gamma$ as a sequence of substitutions of an equivalence class of
terms for each variable in $\Gamma$:
\begin{itemize}
\item the empty sequence is an arrow $\Delta\to\cdot$;
\item Given $f:\Delta\to\Gamma$, type $\istype{\Gamma}{A}$ and term
$\isterm{\Delta}{t}{A\,f}$, where $A\,f$ is the result of applying the substitutions $f$ to
$A$, then $[t/x]\circ f$ modulo equality on $t$ is an
arrow $\Delta\to\Gamma,x:A$;
\item Given $f:\Delta\to\Gamma$ and a well-formed context $\Delta,\lock,
\Delta'$ with no locks in $\Delta'$, then $f$ is also an arrow
$\Delta,\lock,\Delta'\to\Gamma,\lock$;
\end{itemize}
We usually refer to the equivalence classes in arrows via representatives. Note
that substitution respects these equivalence classes because of the congruence
rules.

We next prove that this defines a category. Identity arrows are easily constructed:

\begin{lemma}\label{lem:weak_term_arr}
If $f:\Delta\to\Gamma$ then $f:\Delta,x:A\to\Gamma$.
\end{lemma}
\begin{proof}
By induction on the construction on $f$. The base case is trivial.

Given $f:\Delta\to\Gamma$ and $\isterm{\Delta}{t}{B\,f}$, by induction we
have $f:\Delta,x:A\to\Gamma$ and by variable weakening we have
$\isterm{\Delta,x:A}{t}{B\,f}$ as required.

Supposing we have $f:\Delta\to\Gamma$ yielding $f:\Delta,\lock,\Delta'\to
\Gamma$, we could similarly get $f:\Delta,\lock,\Delta',x:A\to\Gamma$.
\end{proof}

The identity on $\Gamma$ simply replaces all variables by themselves.

\begin{lemma}
The identity on each $\Gamma$ is well defined as an arrow.
\end{lemma}
\begin{proof}
By induction on $\Gamma$. The identity on $\cdot$ is the empty
sequence of substitutions. Given $id:\Gamma\to\Gamma$, we have $id:
\Gamma,x:A\to\Gamma$ by Lemma~\ref{lem:weak_term_arr}, and
$\isterm{\Gamma,x:A}{x}{A}$ as required. $id:\Gamma\to\Gamma$ immediately
yields $id:\Gamma,\lock\to\Gamma,\lock$.
\end{proof}

The composition case is slightly more interesting:

\begin{lemma}\label{Lem:arrow_on_judg}
Given $\Gamma,\Gamma'\vdash\Jud$ and $f:\Delta\to\Gamma$, we have
$\Delta,\Gamma'\,f\vdash\Jud\,f$.
\end{lemma}
\begin{proof}
By induction on the construction on $f$. The base case requires that
$\Gamma'\vdash\Jud$ implies $\Delta,\Gamma'\vdash\Jud$; this \emph{left
weakening} property is easily proved by induction on the typing rules.

Given $f:\Delta\to\Gamma$, $\isterm{\Delta}{t}{A\,f}$ and
$\Gamma,x:A,\Gamma'\vdash\Jud$, by induction
$\Delta,x:A\,f,\Gamma'\,f\vdash\Jud\,f$. Then by
Lemma~\ref{lem:syntactic_sanity} part 3 we have
$\Delta,\subst{(\Gamma'\,f)}{t}{x}\vdash\subst{(\Jud\,f)}{t}{x}$ as required.
The lock case is trivial.
\end{proof}

The composition of $f:\Delta\to\Delta'$ and $g:\Delta'\to\Gamma$ involves
replacing each $[t/x]$ in $g$ with $[t\,f/x]$.

\begin{lemma}
The composition of two arrows $f:\Delta\to\Delta'$ and $g:\Delta'\to\Gamma$
is a well-defined arrow.
\end{lemma}
\begin{proof}
By induction on the definition of $g$. The base case is trivial, and extension by
a new substitution follows via Lemma~\ref{Lem:arrow_on_judg}.


Now suppose we have $g:\Delta'\to\Gamma$ yielding $g:\Delta',\lock,
\Delta''\to\Gamma,\lock$. Now if we have $f:\Delta\to\Delta',\lock,\Delta''$ this
must have arisen via some $f':\Delta_0\to\Delta'$ generating
$f':\Delta_0,\lock,\Delta_1\to\Delta',\lock$, where $\Delta=\Delta_0,\lock,
\Delta_1$. By induction we have well-defined $g\circ f':\Delta_0\to\Gamma$.
Hence $g\circ f':\Delta\to\Gamma,\lock$. But $g\circ f'=g\circ f$ because the
variables of $\Delta''$ do not appear in $g$.
\end{proof}

Checking the category axioms is straightforward.
The category definitions then extend to a CwF in the usual way: the terminal object is
$\Emp$, the families over $\Gamma$ are the types modulo equivalence
well-defined in context $\Gamma$, the elements of any such type are the
terms modulo equivalence, re-indexing is substitution, comprehension
corresponds to extending a context with a new variable, the projection
morphism is the replacement of variables by themselves, and the generic
element is given by the variable rule.

Moving to the definition of a CwDRA, the endofunctor $\L$ acts by mapping
$\Gamma\mapsto\Gamma,\lock$, and does not change arrows. The family
$\R_\Gamma A$ is the type $\istype{\Gamma}{\Box A}$, which is stable under
re-indexing by Lemma~\ref{lem:syntactic_sanity} part 3. The bijections
between families are supplied by the $\shut$ and $\open$ rules, with all
equations following from the definitional equalities.

We do not attempt to prove that the term model is the \emph{initial} CwDRA;
such a result for dependent type theories appears to require syntax be written
in a more verbose style than is appropriate for a paper introducing a new type
theory~\cite{Castellan:Dependent}. Nonetheless our type theory and notion of
model are close enough that we conjecture that such a development is
possible.

\section{A general construction of CwDRAs}
\label{sec:contruction-of-CwDRAs}

In this section we show how to construct a CwDRA from an adjunction of
endofunctors on a category with finite limits. We will refer to categories with
finite limits more briefly as \emph{cartesian} categories.
We will use this construction in Section~\ref{sec:examples} to prove that the examples mentioned in the introduction can indeed be presented as \CwDRAs.
Our construction is an extension of the local universe construction~\cite{LumsdainePL:locumo}, which maps cartesian categories to categories with families, and locally cartesian closed categories to categories with families with $\Pi$- and $\Sigma$-types.
The local universe construction is one of the known solutions to the problem of constructing a strict model of type theory out of a locally cartesian closed category (see 
\cite{hofmann1994interpretation,LumsdainePL:locumo,LumsdainePL:simmuf,hofmann1997syntax} for discussions of alternative approaches to 'strictification').

We first recall the local universe construction.
Since it can be traced back to Giraud's work on fibred
categories~\cite{Giraud:Cohomologie}, we refer to it as the Giraud CwF
associated to a cartesian category.

\begin{definition}
\label{def:Giraud}
  Let $\C$ be a cartesian category. The 
  \textbf{Giraud CwF of $\C$} ($\Gir\C$) is the CwF whose underlying
  category is $\C$, and where a family $A\in\Fam[\Gir\C]{\Gamma}$ is a pair of
  morphisms
    \begin{equation} \label{eq:giraud:family}
      \begin{split}
        \xymatrix{& E\ar[d]^v\\\Gamma  \ar[r]_u & U}
      \end{split}
    \end{equation}
  and an element of $\Elt[\Gir\C]{\Gamma}{A}$, for $A=(u,v)\in\Fam[\Gir\C]{\Gamma}$, is a
  map $a:\Gamma\to E$ such that $v\circ a= u$. Reindexing of $A=(u,v)\in\Fam[\Gir\C]{\Gamma}$ and
    $a\in\Elt[\Gir\C]{\Gamma}{A}$ along
    $\gamma\in\Hom{\Delta}{\Gamma}$ are given by 
    \begin{align} 
      A[\gamma] &\defeq (u\comp \gamma,v)
                  \in\Fam[\Gir\C]{\Delta}\label{eq:8}\\ 
      a[\gamma] &\defeq a\comp \gamma \in\Elt[\Gir\C]{\Delta}{A[\gamma]} 
    \end{align}
  The comprehension $\co{\Gamma}{A}\in\C$, for 
  $A=(u,v)\in\Fam[\Gir\C]{\Gamma}$, is given by the pullback of diagram (\ref{eq:giraud:family}),
      \begin{equation*} 
      \begin{split}
        \xymatrix{\co{\Gamma}{A} \ar[d]_{\proj_A} \ar[r]^{\gen_A} & E\ar[d]^v\\\Gamma  \ar[r]_u & U}
      \end{split}
    \end{equation*}
  with projection morphism $\proj_A$ and generic element $\gen_A$ as indicated in the diagram.
  Note that $\gen_A$ is an element of $A[\proj_A] = (u\circ \proj_A, v)$ as required by commutativity
  of the pullback square. The pairing operation is obtained from the universal property of pullbacks.
\end{definition}

Note that the local universe construction does indeed yield a category with families; in particular, reindexing in $\Gir\C$ is strict as required, simply because reindexing is given by composition.

\begin{remark}\label{Psh}
  The name `local universe' derives from the similarity to Voevodsky's 
  use of a (global) universe $U$ to construct strict models of type 
  theory~\cite{Voevodsky:csys,LumsdainePL:simmuf} in which types in a 
  context $\Gamma$ are modelled as morphisms $\Gamma \to U$. In the local
  universe construction, the universe varies from type to type.   
\end{remark}

In fact, the local universe construction is functorial; a precise statement
requires a novel notion of CwF-morphism:

\begin{definition}\label{def:Gir-map}
  A \define{weak CwF morphism} $\R$ between CwFs consists of a functor $\R : \C\to\D$ between the underlying categories preserving the terminal object, an operation on families mapping $A\in\Fam{\Gamma}$ to a family $\R A\in\Fam[\D]{\R\Gamma}$ and an operation on elements mapping $a\in\Elt{\Gamma}{A}$ to an element $\R a\in\Elt[\D]{\R\Gamma}{\R A}$, such that 
 \begin{enumerate}
\item The functor $\R : \C\to\D$ preserves terminal objects (up to isomorphism)
\item The operations on families and elements commute with reindexing in the sense that 
$\R A [\R \gamma] = \R (A[\gamma])$ and $\R t[\R \gamma] = \R (t[\gamma])$.
\item The maps $\pair{\R \proj_A}{\R \gen_A} : \R(\co{\Gamma}{A}) \to \co{\R\Gamma}{\R A}$ 
are isomorphisms for all $\Gamma$ and $A$. We write $\nu_{\Gamma,A}$ for the inverse.
\end{enumerate}
\end{definition}  

We note the following equalities as consequences of the axioms above. 
\begin{align}
 \R(\proj_A) \circ \nu_{\Gamma,A} & = \proj_{\R A} \label{eq:nu:proj} \\
 \R(\gen_A)[\nu_{\Gamma,A}] & = \gen_{\R A} \label{eq:nu:gen} \\
 \nu_{\Gamma,A}\circ \pair{\R \gamma}{\R a} & = \R\pair{\gamma}a \label{eq:nu:pair} 
\end{align}
For example, the last of these is proved by postcomposing with the inverse of $\nu_{\Gamma,A}$
and noting 
\begin{align*}
  \pair{\R \proj_A}{\R \gen_A}\circ \R\pair{\gamma}a 
  & = \pair{\R \proj_A \circ \R\pair{\gamma}a}{\R \gen_A[\R\pair{\gamma}a]} \\
  & = \pair{\R (\proj_A \circ \pair{\gamma}a)}{\R (\gen_A[\pair{\gamma}a])} \\
  & = \pair{\R \gamma}{\R a}
\end{align*}

Note that a weak CwF morphism preserves comprehension and the terminal object only up to isomorphism instead of on the nose, as required by the
stricter notion of morphism of Dybjer~\cite[Definition 2]{Dybjer1996}.
Weak CwF morphisms sit between strict CwF-morphisms and pseudo-CwF morphisms~\cite{DBLP:journals/lmcs/CastellanCD17}. The latter allow substitution to be 
preserved only up to isomorphism satisfying a number of coherence conditions. Since weak
CwF morphisms preserve substitution on the nose, these are not needed here. 

 
\begin{theorem}
  \label{functor-lex}
  $\Gir$ extends to a 
  functor from the category of cartesian categories and
  finite limit preserving functors, to the category of CwFs with weak morphisms.
\end{theorem}
\begin{proof}
  Let $\R:\C\to\D$ be a finite limit preserving functor.
  For each $\Gamma\in\C$ and $A=(u,v)\in\Fam[\Gir\C]{\Gamma}$,
  we simply let $\R A \defeq (\R u,\R v)$. Likewise, for an element
  $a\in\Elt[\Gir\C]{\Gamma}{A}$, we let $\R a$ be the
  action of $\R$ on the morphism $a$. 
    Finally, since comprehension is defined by pullback and $\R$ preserves pullbacks up to isomorphism, we obtain the required $\nu_{\Gamma,A}$.
\end{proof}

We now embark on showing that if we apply the local universe construction to a cartesian category $\C$ with a pair of adjoint endofunctors, then the resulting CwF $\Gir\C$ is in fact a \CwDRA\ (Theorem~\ref{thm:giraud}).
To this end, we introduce the auxiliary notion of a category with families
with an adjunction:

\begin{definition}\label{def:CwF+A}
  A \define{CwF+A} consists of a CwF with an adjunction $\L\adj \R$ on the category of contexts, such that $\R$ extends to a weak CwF endomorphism.
\end{definition}

\begin{lemma}\label{CwDRAfromCwF+A}
  If $\C$ with the adjunction $\L\dashv \R$ is a CwF+A, then there is a \CwDRA\ structure on $\C$ with $\L$ as the required functor on $\C$.
\end{lemma}
\begin{proof}
  We write $\eta$ for the unit of the adjunction.
  For a family $A\in\Fam{\L \Gamma}$, we define $\R_{\Gamma} A \in\Fam{\Gamma}$ to be $(\R A)[\eta]$.
  For an element $a\in\Elt{\L \Gamma}{A}$, we define its transpose $\transp{a}\in\Elt{\Gamma}{\R_\Gamma A}$ to be 
  $(\R a)[\eta]$.
  For the opposite direction, suppose $b\in\Elt{\Gamma}{\R_\Gamma A}$.
  Since $\pair{\eta}{b}:\Gamma \rightarrow \co{\R\L\Gamma}{\R A}$, we have that $\L(\nu_{\L\Gamma,A} 
  \circ \pair{\eta}{b}):\L\Gamma\rightarrow \L\R(\co{\L\Gamma}{A})$
  and thus we can define $\transp{b}\in\Elt{\L\Gamma}{A}$ to be the element $\gen_A[\varepsilon\circ \L(\nu_{\L\Gamma,A} 
  \circ \pair{\eta}{b})]$.
  Note that this is well typed because $\gen_A$ is an element of the family $A[\proj_A]$ and so $\transp b$ is an element of
\begin{align*}
  A[\proj_A \circ \varepsilon\circ \L(\nu \circ \pair{\eta}{b})] & =  A[\varepsilon\circ \L(\R\proj_A \circ \nu \circ \pair{\eta}{b})] \\
  & =  A[\varepsilon\circ \L(\proj_{\R A} \circ \pair{\eta}{b})] \\
  & =  A[\varepsilon\circ \L(\eta)] \\
  & = A
\end{align*}
using equation (\ref{eq:nu:proj}) in the second equality. These operations can be proved inverses of each other using the 
equations (\ref{eq:nu:gen}) and (\ref{eq:nu:pair}).
%
\end{proof}

Note that the conditions for a CwF+A are stronger than those for a \CwDRA; for
instance, a \CwDRA\ does not require $\R$ to be defined on the context
category.  We return to the relation between these constructions in
Section~\ref{sec:CwFAfromCwDRA}

\begin{lemma}\label{thm:giraud-CwF+A}
  If $\C$ is a cartesian category and $\L\dashv \R$ are adjoint endofunctors on $\C$, then $\Gir\C$ with the adjunction $\L\dashv \R$ is a CwF+A.
\end{lemma}
\begin{proof}
  We are already given an adjunction on the underlying category of $\Gir\C$.
  Theorem~\ref{functor-lex} constructs the weak CwF morphism.
\end{proof}

\begin{theorem}
  \label{thm:giraud}
  If $\C$ is a cartesian category and $\L\dashv \R$ are adjoint
  endofunctors on $\C$, then $\Gir\C$ has the structure of a \CwDRA.
\end{theorem}
\begin{proof}
By Lemmas~\ref{thm:giraud-CwF+A} and~\ref{CwDRAfromCwF+A}.
\end{proof}

The above Theorem~\ref{thm:giraud} thus provides a general construction of \CwDRAs.
In Section~\ref{sec:examples} we use it to present examples from the literature.
As mentioned earlier, the local universe construction interacts well with other type formers: If we start with a locally cartesian closed category $\C$ (with W-types, Id-types and a universe), then $\Gir\C$ also models dependent products $\Pi$ and sums $\Sigma$ (and W-types, Id-types and a universe); see~\citeasnoun{LumsdainePL:locumo}.
In Section~\ref{sec:universes} we consider universes.

\subsection{CwF+A from a CwDRA}\label{sec:CwFAfromCwDRA}

In this subsection we show how to produce a CwF+A from a \CwDRA\ under the assumption that the CwF is \emph{democratic}.
Intuitively, a democratic CwF is one where every context comes from a type, and hence it is not surprising that for a democratic CwDRA one can use the action of the dependent right adjoint on families to define a right adjoint on contexts.

\begin{definition}
  A CwF is \define{democratic}~\cite{Clairambault2011} if for every context $\Gamma$ there is a family $\widehat{\Gamma}\in\Fam{\Terminal}$ and an isomorphism $\zeta_\Gamma:\Gamma \rightarrow \co{\Terminal}{\widehat{\Gamma}}$.
\end{definition}

\begin{theorem}\label{thm:giraud-converse}
  Let $\C$ be a democratic \CwDRA. The endofunctor ${\L}:\C\to\C$, part of the \CwDRA\ structure,
  has a right adjoint $\R$.
\end{theorem}
 \begin{proof}
For $\Gamma\in\C$, we define $\R\Gamma\in\C$ by 
\begin{equation}
    \label{eq:17}
    \R \Gamma \defeq \co{\Terminal}{\R_\Terminal(\widehat{\Gamma}[!_{\L \Terminal}])}
  \end{equation}
  We have a bijection, natural in $\Delta$
  \begin{align*} \C(\Delta, \R \Gamma) & \cong \Elt{\Delta}{\R_\Terminal(\widehat{\Gamma} [!_{\L \Terminal}])) [!_\Delta])} \\
  &\cong \Elt{\Delta}{\R_\Delta (\widehat{\Gamma} [!_{\L \Delta}]))} \\
  &\cong \Elt{\L \Delta} {\widehat{\Gamma} [!_{\L \Delta}])}\\
  &\cong \C(\L \Delta, \co{\Terminal}{\widehat{\Gamma})}\\
  &\cong \C(\L\Delta, \Gamma)
  \end{align*}
  The last of the above bijections follows by composition with $\inv{\zeta_\Gamma}$.
  
  Let $\gamma:\Gamma' \to \Gamma$ we have then an action $ \gamma^*:\C(-,\R\Gamma')\to \C(-,\R \Gamma)$ given by
  \[\C(-,\R \Gamma') \cong \C(\L-,\Gamma') \xrightarrow{-\circ \gamma} \C(\L-,\Gamma) \cong \C(-,\R \Gamma)\]
  Define $\R \gamma = \gamma^*_{\R \Gamma'} (\id_{\R \Gamma'})$.
  Then the correspondence $\C(\Delta, \R \Gamma) \cong \C(\L\Delta, \Gamma)$ is natural in $\Gamma$, proving that 
  $\R$ is a right adjoint to $\L$. 
\end{proof}   
Consider a democratic \CwDRA, with $\C$ as the underlying category, and $\L\dashv \R$ the adjunction 
obtained from the above theorem. We then extend $\R$ to a weak CwF morphism by
defining, for a family  $A \in\Fam{\Gamma}$ and an element $a\in\Elt{\Gamma}{A}$,
\begin{align*}
  &\R A \defeq \R_{\R \Gamma} (A[\varepsilon]) &\R a \defeq \transp{a[\varepsilon]}
\end{align*}
where $\varepsilon:\L \R \Gamma \to \Gamma$ is the counit of the adjunction.
\begin{lemma}
\label{lem:R_w_morphism}
  $\R$ as defined above is a weak CwF morphism. In particular, for $A \in \Fam{\Gamma}$
  we have an isomorphism $\nu_{\Gamma,A}:\co{\R \Gamma}{\R A} \to \R(\co{\Gamma}{A})$, inverse to $\pair{\R \proj_A}{\R \gen_A}$.
\end{lemma}

\begin{proof}
We will show a bijection $\C(\Delta, \R \Gamma.\R A) \cong \C(\Delta, \R(\Gamma.A))$ natural in $\Delta$. We have
\begin{align*}
\C(\Delta,\R \Gamma.\R A) &\cong \prod_{\gamma:\C(\Delta,\R \Gamma)} \Elt{\Delta}{(\R A) [\gamma]}
\end{align*}
We have a bijection $-^\top: \C(\Delta,\R \Gamma) \cong \C(\L\Delta,\Gamma)$. But 
\[(\R A)[\gamma] = (\R_{\R \Gamma} A[\varepsilon])[\gamma] = \R_\Delta (A[\varepsilon \circ \L \gamma]) = \R_\Delta (A[\gamma^\top])\] 
Hence we have a bijection $\Elt{\Delta}{(\R A)[\gamma]} \cong \Elt{\L \Delta}{A[\gamma^\top]}$. So
\begin{align*}
\C(\Delta,\R \Gamma.\R A) &\cong \prod_{\gamma:\C(\Delta,\R \Gamma)} \Elt{\Delta}{(\R A) [\gamma]}\\
& \cong \prod_{\gamma':\C(\L\Delta,\Gamma)} \Elt{\L\Delta}{A[\gamma']}\\
& \cong \C(\L \Delta, \Gamma.A) \\
& \cong \C(\Delta,\R(\Gamma.A)) 
\end{align*}
By the Yoneda lemma, this implies $\R \Gamma.\R A \cong \C(\Delta,\R(\Gamma.A))$, and it is easy to check that the  
direction $\C(\Delta,\R(\Gamma.A)) \to \R \Gamma.\R A$ is given by $\pair{\R \proj_A}{\R \gen_A}$. 
\end{proof}

\begin{corollary}
\label{cor:Dem_cwFDRA_is_CwF+A}
A democratic CwDRA has the structure of CwF+A
\end{corollary}

\begin{remark}
  For a category $\C$ with a terminal object, the CwF $\Gir\C$ is democratic with $\widehat{\Gamma}$ given by the diagram:
  \[
    \xymatrix{& \Gamma\ar[d]^{!_\Gamma}\\{1}\ar[r]_{!_{1}} & 1}
  \]
\end{remark}
\begin{remark}
  For ordinary dependent type theory, the term model is a democratic CwF~\cite[Section 4]{DBLP:journals/lmcs/CastellanCD17}.
  However, the term model for our modal dependent type theory is \emph{not} democratic, since there is, for example, no type corresponding to the context $\lock$ consisting of just one lock. 
\end{remark}


\section{Examples}
\label{sec:examples}

We now present concrete examples of \CwDRA s generated from
cartesian categories with an adjunction of endofunctors, including those
mentioned in the introduction.

\medskip\noindent\textbf{$\Pi$ type with closed domain}\quad
Consider a CwF where the underlying category of contexts $\C$ is cartesian closed, and let $A$ be a closed type. 
We have then an adjunction of endofunctors ${-\times \co{\Terminal}{A} \adj -^{\co{\Terminal}{A}}}$ on 
$\C$, and suppose that the right adjoint
extends to a weak CwF endomorphism, giving the structure of a CwF+A. As we saw above, this happens e.g. when the 
CwF is of the form $\Gir\C$. In this case $\R_\Gamma B$ behaves as a type of the form $\Pi(x:A) B$ since 
$\C(\Gamma \vdash \R_\Gamma B) \cong \C(\Gamma \times \co\Terminal A \vdash B)
 \cong \C(\co\Gamma{(A[!_{\Gamma}])} \vdash B)$.

Thus, the notion of dependent right adjoint generalises $\Pi$ types with closed domain. This generalises to the setting where $\C$ carries the structure of a monoidal closed category, in which case the 
adjunction $-\otimes \co{\Terminal}{A} \adj {\co{\Terminal}{A}} \multimap (-)$ extends to give a dependent notion of linear
function space with closed domain. The next example is an instance of this. 

\medskip\noindent\textbf{Dependent name abstraction}\quad
The notion of \emph{dependent name abstraction} for families of nominal sets was introduced by Pitts et al.~\cite[Section~3.6]{PittsAM:deptta} to give a semantics for an extension of Martin-L\"of Type Theory with names and constructs for freshness and name-abstraction.
It provides an example of a \CwDRA\ that can be presented via Theorem~\ref{thm:giraud}.
In this case $\C$ is the category $\Nom$ of nominal sets and equivariant functions~\cite{PittsAM:nomsns}.
Its objects are sets $\Gamma$ equipped with an action of finite permutations of a fixed infinite set of atomic names $\Atom$, with respect to which the elements of $\Gamma$ are finitely supported, and its morphisms are functions that preserve the action of name permutations.
$\Nom$ is a topos (it is equivalent to the Schanuel topos~\cite[Section~6.3]{PittsAM:nomsns}) and hence in particular is cartesian.
We take the functor $\L:\Nom\rightarrow\Nom$ to be separated product~\cite[Section~3.4]{PittsAM:nomsns} with the nominal set of atomic names.
This has a right adjoint $\R$ that sends each $\Gamma\in\Nom$ to the nominal set of name abstractions $[\Atom]\Gamma$~\cite[Section~4.2]{PittsAM:nomsns} whose elements are a generic form of $\alpha$-equivalence class in the case that $\Gamma$ is a nominal set of syntax trees for some language.

Applying Theorem~\ref{thm:giraud}, we get a \CwDRA\ structure on $\Gir\Nom$.
In fact the CwF $\Gir\Nom$ has an equivalent, more concrete description in this case, in terms of \emph{families of nominal sets}~\cite[Section~3.1]{PittsAM:deptta}.
Under this equivalence, the value $\R_\Gamma A\in\Gir\Nom(\Gamma)$ of the dependent right adjoint at $A\in\Gir(\L\Gamma)$ corresponds to the family of \emph{dependent name abstractions} defined by~\citeasnoun[Section~3.6]{PittsAM:deptta}.
The bijection \eqref{eq:18} is given in one direction by the name abstraction operation~\cite[(40)]{PittsAM:deptta} and in the other by concretion at a fresh name~\cite[(42)]{PittsAM:deptta}.

\medskip\noindent\textbf{Guarded and Clocked Type Theory}\quad
Guarded recursion~\cite{Nakano:Modality} is an extension of type theory with a modal \emph{later} operator, denoted $\later$, on types, an operation $\nxt : A \to \later A$ and a guarded fixed point operator $\fix : (\later A \to A) \to A$ mapping $f$ to a fixed point for $f \circ \nxt$.
The standard model of guarded recursion is the topos of trees~\cite{birkedal2011first}, i.e. the category of presheaves on $\omega$, with $\later X(n+1) = X(n)$, $\later X(0) = 1$.
The later operator has a left adjoint $\earlier$, called \emph{earlier}, given by $\earlier X(n) = X(n+1)$, so $\later$ yields a dependent right adjoint on
the induced \CwDRA.

Birkedal et al.~\cite[Section 6.1]{birkedal2011first} show that $\earlier$ in a
dependently typed setting does not commute with reindexing. However it
\emph{does} have a left adjoint, namely the `stutter' functor $!$ with
$!X(0)=X(0)$ and $!X(n+1)=X(n)$, so $\earlier$ does give rise to a
well-behaved modality in the setting of this paper. This apparent contradiction is
resolved by the use of locks in the context: $\istype{\Gamma}{A}$ does not
give rise to a well-behaved $\istype{\Gamma}{\earlier A}$, but
$\istype{\Gamma,\lock}{A}$ does. This is an intriguing example of the
Fitch-style approach increasing expressivity.

%
%


Guarded recursion can be used to encode coinduction given a \emph{constant}
modality \cite{clouston2015programming}, denoted $\Box$, on the topos of trees, defined as $\Box X(n) = \lim_k X(k)$.
The $\Box$ functor is the right adjoint of the essential geometric morphism on $\hat{\omega}$ induced by $0 : \omega \to \omega$, the constant map to $0$, and hence it also yields a dependent right adjoint.
In~\citeasnoun{clouston2015programming}, $\Box$ was used in a
simple type theory, employing `explicit substitutions'
following~\citeasnoun{Bierman:Intuitionistic}.
As we will discuss in Section~\ref{sec:discussion} this approach proved difficult
to extend to dependent types, and we wish to use the modal dependent type
theory of the present paper to study $\Box$ in dependent type theory.

An alternative to the constant modality are the \emph{clock quantifiers}
of~\citeasnoun{atkey13icfp}, which unlike the constant modality 
have already been combined succesfully with dependent
types~\cite{Mogelberg:Type,BirkedalL:gdtt-conf}. They are also slightly more
general than the constant modality, as multiple clocks allow coinductive data
structures that unroll in multiple dimensions, such as infinitely-wide infinitely-deep
trees. The denotational semantics, however, are more
complicated, consisting of presheaves over a category of `time objects',
restricted to those fulfilling an `orthogonality' condition~\cite{GDTTmodel}. Nevertheless the
$\earlier \dashv \later$ adjunction of the topos of trees lifts to this category,
and so once again we may construct a CwDRA.

Clocked Type Theory (\clott)~\cite{bahr2017clocks} is a recent type theory for
guarded recursion that has strongly normalising reduction semantics, and
has been shown to have semantics in the category discussed
above~\cite{CloTTmodel}.
The operator $\later$ is refined to a form of dependent function type
$\latbind \tickA\kappa A$ over ticks $\tickA$ on clock $\kappa$.
Ticks can appear in contexts as $\Gamma, \tickA : \kappa$; these are similar to
the locks of Fitch-style contexts, except that ticks have names, and can be
weakened.
The names of ticks play a crucial role in controlling fixed point unfoldings.

Finally, the modal operator $\later$ on the topos of trees can be generalized to
the presheaf topos $\widehat{\C\times \omega}$ for any category $\C$,
simply by using the identity on $\C$ to extend the underlying functor (which
generates the essential geometric morphism) on $\omega$ to $\C\times \omega$.
In~\citeasnoun{GCTT} this topos, with $\C$ the cube category, is used to model guarded \emph{cubical} type theory; an extension of cubical type theory~\cite{CoquandT:cubttc}.
In more detail, one uses a CwF where families are certain \emph{fibrations}, and since $\later$ preserves fibrations, it does indeed extend to a \CwDRA.


\medskip\noindent\textbf{Cohesive Toposes}\quad
Cohesive toposes have also recently been considered as models of a form of modal type theory~\cite{shulman2018brouwer,2017arXiv170607526R}.
Cohesive toposes carry a triple adjunction $\int\dashv\flat\dashv\sharp$ and hence induce two dependent right adjoints.
Examples of cohesive toposes include simplicial sets $\hat{\Delta}$ and cubical sets $\hat{\Box}$; since these are presheaf toposes they also model universes.
For example, for simplicial sets, the triple of adjoints are given by the essential geometric morphism induced by the constant functor $0:\Delta\to \Delta$. In
the category of cubical sets $\sharp$ has a further right adjoint, used
by~\citeasnoun{Nuyts:Parametric} to reason about parametricity.

\medskip\noindent\textbf{Tiny objects}\quad
\citeasnoun{licata2018internal} use a `tiny' object $\mathbb{I}$ to
construct the fibrant universe in the cubical model of homotopy type
theory.  By definition, an object $\mathbb{I}$ of a category $\C$ is
\emph{tiny} if the exponentiation functor
$(-)^{\mathbb{I}}:\C\rightarrow\C$ has a right-adjoint, which they
denote by $\surd$. As for $\earlier$ above, the right adjoint functor
$\surd$ exists globally, but not locally; in other words, there is no
right adjoint to $(-)^{\top.\mathbb{I}}$ on each category of families
over an object $\Gamma\in\C$, stable under re-indexing
$\Gamma$ (except in the trivial case that $\mathbb{I}$ is
terminal). Nevertheless our present framework is still applicable: the
corresponding \emph{dependent} right adjoint for $(-)^{\mathbb{I}}$,
constructed as in Section~\ref{sec:contruction-of-CwDRAs}, plays an
important part in the construction of the fibrant universe given in
\cite{licata2018internal}.

\section{Universes}
\label{sec:universes}

In this section, we extend our modal dependent type theory with
universes.  For the semantics, we start from Coquand's notion of a
category with universes~\cite{Coquand:CwU}, which covers all presheaf
models of dependent type theory with universes. The notion of
\emph{category with universes} rests on the observation that in
presheaf models one can interpret an inverse $\code{-}$ to the usual
function $\El$ from codes to types, and hence obtain a simpler notion
of universe than
usual~\citeaffixed[section~2.1.6]{hofmann1997syntax}{such as in}.

\begin{definition}[category with universes]
  A \CwU\ is specified by:
  \begin{enumerate}
  \item A category $\C$ with a terminal object $\Terminal$.
  \item For each object $\Gamma\in\C$ and natural number $n\in\nats$, a set $\FamU{\Gamma}{n}$ of
    \emph{families at universe level $n$} over $\Gamma$.
  \item For each object $\Gamma\in\C$, natural number $n$, and family $A\in\FamU{\Gamma}{n}$, a
    set $\Elt{\Gamma}{A}$ of \emph{elements} (at some level) of the family $A$ over
    $\Gamma$.
  \item For each morphism
    $\gamma\in\Hom{\Delta}{\Gamma}$, \emph{re-indexing}
    functions $A\in\FamU{\Gamma}{n} \mapsto A[\gamma]\in\FamU{\Delta}{n}$ and
    $a\in\Elt{\Gamma}{A}\mapsto a[\gamma]\in\Elt{\Delta}{A[\gamma]}$,
    satisfying equations for associativity and identity as in a CwF.
  \item For each object $\Gamma\in\C$, number $n$ and family
    $A\in\FamU{\Gamma}{n}$, a \emph{comprehension object}
    $\co{\Gamma}{A}\in\C$ equipped with projections and generic elements
    satisfying equations as in a CwF.
  \item For each number $n$, a family $\U_{n}\in\FamU{\Terminal}{n+1}$, the
    \emph{universe at level $n$}.
  \item For each object $\Gamma\in\C$ and number $n$,
    a \emph{code} function $A\in\FamU{\Gamma}{n}\mapsto \code{A}\in\Elt{\Gamma}{\U_{n}[!_{\Gamma}]}$,
    and an \emph{element} function $u\in\Elt{\Gamma}{\U_{n}[!_{\Gamma}]} \mapsto
    \E u \in\FamU{\Gamma}{n}$, satisfying 
    $\code{A}[\gamma] = \code{A[\gamma]}$, 
    $\E \code{A} = A$, and $\code{\E u} = u$.
  \end{enumerate}
\end{definition}

We will of course want the universes to be closed under various
type-forming operations, but in this formalisation of universes
these definitions are just as for CwFs,
without having to explicitly reflect them into the universes.

\begin{lemma}\label{lem:1}
  The element function is stable under re-indexing:
  $(\E u)[\gamma] = \E(u[\gamma])$.
\end{lemma}
\begin{proof}
  $(\E u)[\gamma] = \E\code{(\E u)[\gamma]} = \E(\code{\E u}[\gamma])
  = \E(u[\gamma])$.
\end{proof}

\begin{corollary}
\label{cor:generic family CwU}
  In a CwU there is a \emph{generic family} $\El\in\FamU{\co{\Terminal}{\U_{n}}}{n}$
  of types of level $n$ (for each $n\in\nats$), 
  with the property that 
  $\El[\pair{!_{\Gamma}}{\code{A}}] = A$, for all
  $A\in\FamU{\Gamma}{n}$.
\end{corollary}
\begin{proof}
  Since $\proj_{\U_n} =\ ! :\co{\Terminal}{}{\U_n}\to\Terminal$, 
  we have $\gen\in\Elt{\co{\Terminal}{\U_n}}{\U_n[!_{\co{\Terminal}{}{\U_n}}]}$
  and thus we can define $\El$ to be $\E\gen$, and then the required property follows
  by Lemma~\ref{lem:1}.
\end{proof}


For a CwU, there is an underlying CwF with families over $\Gamma$ given as 
$\Fam\Gamma = \bigcup_n \FamU\Gamma n$. 
Using this we can extend the definition of CwDRA to categories with
universes in the obvious way, as follows:

\begin{definition}[CwUDRA]
A \define{category with universe and dependent right adjoint (\CwUDRA)} is a CwU with the structure of a 
\CwDRA\ such that operation on types preserves universe levels in the sense that $A \in \FamU{\L \Gamma}n$ implies
$\R_\Gamma A \in \FamU\Gamma n$. 
\end{definition}

Similarly, one can extend the notion of CwF+A from Definition~\ref{def:CwF+A}
to the setting of universes:

%
\begin{definition}[CwU+A]
  A \define{weak CwU morphism} $\R$ is a weak CwF morphism on the underlying CwFs preserving size in the sense that 
  $A \in \FamU{\Gamma}n$ implies
$\R A \in \FamU{\R\Gamma}n$.
  A \define{CwU+A} consists of a CwU with an adjunction $\L\dashv\R$ on the category of contexts, such that $\R$ extends to a weak CwU morphism.
\end{definition}

The construction of Lemma~\ref{CwDRAfromCwF+A} extends to a construction of a CwUDRA from a
CwU+A. We now show (Lemma~\ref{lem:univ-endo-weak-morphism}) that the action of the right adjoint
on families and elements can be defined by just defining it on the universe as in the following definition.

\begin{definition}\label{def:amp}
  A \define{universe endomorphism} on a CwU is a finite limit
  preserving functor $\R$ on the category of contexts together with,
  for each $n$, a family $\Rl\in\FamU{\R(\co{\Terminal}{\U_n})}{n}$
  and an element
  $\r\in\Elt{\R(\co{\co{\Terminal}{\U_n}}{\El})}{\Rl[\R \proj]}$ such
  that the morphism
  \begin{equation}\label{eq:5}
    \xymatrix{{\R(\co{\co{\Terminal}{\U_n}}{\El})}
      \ar[rr]^{\pair{\R\proj}{\r}} \ar[dr]_{\R\proj} &&
      {\co{\R(\co{\Terminal}{\U_n})}{\Rl}} \ar[dl]^{\proj}\\
      & {\R(\co{\Terminal}{\U_n})}}
  \end{equation}
  over $\R(\co{\Terminal}{\U_n})$ is an isomorphism; in other words
  there is a morphism
  $\l: \co{\R(\co{\Terminal}{\U_n})}{\Rl} \to
  \R(\co{\co{\Terminal}{\U_n}}{\El})$ satisfying
  $\l\comp\pair{\R\proj}{\r} = \id$ and
  $\pair{\R\p}{\r}\comp\l = \id$.
\end{definition}

This means that we have a universe category endomorphism in the sense
of~\citeasnoun[Section~4.1]{Voevodsky:csys}: a family
$\Rl\in\FamU{\R(\co{\Terminal}{\U_n})}{n}$ gives a pullback square
with the morphism $\co{\co{\Terminal}{\U_n}}{\El}\to
\co{\Terminal}{\U_n}$ and the code function. The isomorphism above implies that the universe
$\R(\co{\co{\Terminal}{\U_n}}{\El}) \to \R(\co{\Terminal}{\U_n})$ is also
pullback of $\co{\co{\Terminal}{\U_n}}{\El}\to
\co{\Terminal}{\U_n}$ along the code function.

%

Given a CwU with a weak CwU morphism $\R$, then clearly $\R$ is a universe endomorphism, with $\Rl\defeq \R(\El)$, $\r\defeq \R\gen$ and $\l\defeq \nu$.
Conversely:

\begin{lemma}
  \label{lem:univ-endo-weak-morphism}
  Any CwU with a universe endomorphism $\R : \C \to \C$ extends to a weak CwU morphism.
\end{lemma}
\begin{proof}
    Given $A\in\C(\Gamma,n)$, since we have
  $\pair{!_{\Gamma}}{\code{A}}:\Gamma\fun\co{\Terminal}{\U_n}$, we can define
  \begin{equation}
    \label{eq:3}
    \R A \defeq \Rl[\R\pair{!_{\Gamma}}{\code{A}}] \in \FamU{\R\Gamma}n
  \end{equation}
  This is stable under re-indexing, since for $\gamma:\Delta\to \Gamma$
  \begin{align*}
     \R(A\,[\gamma]) & \defeq \Rl[\R\pair{!_{\Delta}}{\code{A\,[\gamma]}}] \\
     &= \Rl[\R\pair{!_{\Delta}}{\code{A}[\gamma]}] \\
     & = \Rl[\R(\pair{!_{\Gamma}}{\code{A}}\comp\gamma)] \\
     & = \Rl[\R\pair{!_{\Gamma}}{\code{A}} \comp\R\gamma] \\
     & = (\Rl[\R\pair{!_{\Gamma}}{\code{A}}])[\R\gamma] \\
     & \defeq (\R A)[\R\gamma]
  \end{align*}

  Given $a\in \Elt\Gamma A$, by Corollary~\ref{cor:generic family CwU} we have
  $a \in \Elt\Gamma{\El[\pair{!_{\Gamma}}{\code{A}}]}$ and hence
  \[\pair{\pair{!_{\Gamma}}{\code{A}}}{a}:\Gamma\fun\Terminal.\U_n.\El\]
  Therefore
  \[
    \r[\R \pair{\pair{!_{\Gamma}}{\code{A}}}{a}] \in\Elt{\R\Gamma}{(\Rl[\R\p])[\R
    \pair{\pair{!_{\Gamma}}{\code{A}}}{a}]}
  \]
  But
  $(\Rl[\R\p])[\R \pair{\pair{!_{\Gamma}}{\code{A}}}{a}] = \Rl[\R(\p\comp
  \pair{\pair{!_{\Gamma}}{\code{A}}}{a})] = \Rl[\R\pair{!_{\Gamma}}{\code{A}}) \defeq
  \R A$. We can therefore define
  \begin{equation}
    \label{eq:4}
    \R a \defeq \r[\R \pair{\pair{!_{\Gamma}}{\code{A}}}{a}] \in \Elt{\R\Gamma}{\R A}
  \end{equation}
  and this is stable under re-indexing, since for $\gamma:\Delta \to \Gamma$
  \begin{align*}
      (\R a)[\R \gamma] &\defeq \r[\R \pair{\pair{!_{\Gamma}}{\code{A}}}{a}] [\R
    \gamma] \\&= \r[\R \pair{\pair{!_{\Delta}}{\code{A}[\gamma]}}{a[\gamma]}] \\&=
    \r[\R \pair{\pair{!_{\Delta}}{\code{A[\gamma]}}}{a[\gamma]}] \\
    &\defeq\R(a[\gamma])
  \end{align*}
  
%
Finally we must show that $\R$ commutes with comprehension. For this, note that there are pullback squares
  \[
    \xymatrix{{\R(\co\Gamma A)} \pullbackcorner
      \ar[rrr]^{\R\pair{\pair{!_{\co\Gamma A}}{\code{A}[\proj_A]}}{\gen}} \ar[d]_{\R\proj} 
      &&&  \R(\co{\co{\Terminal}{\U_n}}{\El}) \ar[d]^{\R\p}\\
      {\R\Gamma}\ar[rrr]_{\R\pair{!_{\Gamma}}{\code{A}}} &&&
      {\R(\Terminal.\U_n)}}
    \qquad
    \xymatrix{{\R\Gamma.\R A} \pullbackcorner
      \ar[rrr]^{\pair{\R\pair{!_{\Gamma}}{\code{A}}\circ\proj}{\gen}} \ar[d]_\p &&&
      {\R(\Terminal.\U_n).\Rl} \ar[d]^\p\\
      {\R\Gamma}\ar[rrr]_{\R\pair{!_{\Gamma}}{\code{A}}} &&&
      {\R(\Terminal.\U_n)}}
  \]
  the former because the functor $\R$ preserves finite limits and the
  latter by definition of $\R A$. Applying \eqref{eq:4} with $a=\gen$
  we get that the pullback along $\R\pair{!_{\Gamma}}{\code{A}}$ of
  the morphism $\pair{\R\proj}{\r}$ in \eqref{eq:5} is
  \[
    \xymatrix{{\R(\Gamma. A)} \ar[rr]^{\pair{\R\proj}{\R\gen}}
      \ar[dr]_{\R\proj} &&
      {\R\Gamma.\R A} \ar[dl]^{\proj} \\
      & {\R\Gamma}}
  \]
  Then since $\pair{\R\proj}{\r}$ is an isomorphism, so is its
  pullback $\pair{\R\proj}{\R\gen}$, as required.
\end{proof}

\begin{remark}\label{sec:codes-image-r}
  We observe that for $\R$ as constructed above, the image under $\R$ of maps with $\U_n$-small fibers is classified by $\Rl\in\FamU{\R(\co{\Terminal}{\U_n})}{n}$. That is to say that $\pair {!_{\R \Gamma}} {\code{\R A}}= \pair{!_{\R\Gamma}} {\code{\Rl}} \circ \R \pair{!_\Gamma}{\code{A}}$ which is true by our choice of $\R A = \Rl [\R \pair{!_\Gamma}{\code{A}}]$. 
  Hence, the type of codes for such fibers is $\R$ applied to the codes for types.
  The same situation occurred for $\later$ in~\citeasnoun[V.5]{DBLP:conf/lics/BirkedalM13}, but was not observed at the time.
\end{remark}

\begin{theorem}
  \label{thm:cwudra}
  Any CwU equipped with an adjunction on the category of contexts whose
  right adjoint is a universe endomorphism can be given the structure
  of a CwUDRA.
\end{theorem}
\begin{proof}
  Combine Lemmas~\ref{CwDRAfromCwF+A} and \ref{lem:univ-endo-weak-morphism}.
\end{proof}

For most of the presheaf examples considered in Section~\ref{sec:examples}, the dependent right adjoint is
obtained as the direct image of an essential geometric morphism
arising from a functor on the category on which the presheaves are
defined. 
We show that in this case, the right adjoint preserves universe levels
and hence gives a CwUDRA. 
For simplicity, we will restrict to one universe and show that the right adjoint
preserves smallness with respect to this. 

Let $U$ be a universe in an ambient set theory. We call the elements
of $U$, $U$-sets. A \emph{$U$-small category} is one where both the
sets of objects and the set of morphisms are $U$-small. 
Let us assume that $U$ is $U$-complete ---
it is closed under limits of $U$-small diagrams. A Grothendieck
universe in ZFC would satisfy these conditions. 


\begin{proposition}
Let $C,D$ be a $U$-small categories and $f:C\to D$ a functor between
them. The direct image $f_*$ of the induced geometric morphism
preserves size. In particular, for each endofunctor $f$, the direct image is a weak CwU morphism.
\end{proposition}
\begin{proof}
Since $f_*$ is a right adjoint, we know that it induces a weak CwF morphism, and we
just need to show that it maps $U$-small families to $U$-small families. Recall first that 
the direct image $f_*$ is the (pointwise) right Kan
extension~\cite[A4.1.4]{johnstone:elephant} defined on
objects by the limit of the diagram \[
(\Ran_f F)d\defeq\lim (f\downarrow d) \stackrel{\pi_1}\to
C^{op}\stackrel{F}\to Uset,
\]
for $F\in \widehat{C}$ and $d\in D$. Here $(f\downarrow d)$ denotes the
\emph{comma category} consisting of pairs $(c;g:f(c)\to d)$. 

A family $\alpha:F\to G$, for $F,G\in \hat{C}$ is $U$-small if for each 
$c$ and each $x \in G(c)$ the set $\inv{\alpha_c}(x)$ is in $U$. Given
$(x_g)_{g \in f\downarrow d} \in f_*G(d) = (\Ran_f G)d$, the preimage $\inv{(f_*\alpha)}((x_g)_{g\in f\downarrow d})$ 
is the set
\[
\{(y_g)_{g\in f\downarrow d} \in (\Ran_f G)d \mid \forall g . \alpha_c(y_g) = x_g\}
\] 
which is the limit of the diagram associating to each $g$ the set 
$\inv\alpha_c(x_g)$. Since each of these sets are in $U$ by assumption 
and since also $f\downarrow d$  is in $U$, by the assumption of $U$ being closed under limits, 
also $\inv{(f_*\alpha)}((x_g)_{g\in f\downarrow d})$ is in $U$ as desired. 
\end{proof}

\medskip\noindent\textbf{Syntax}\quad At this stage it should hopefully be clear that one can refine and extend the syntax of modal dependent type theory from Section~\ref{sec:syntax} so that the resulting syntactic type theory can be modelled in a CwU+A.
The idea is, of course, to refine the judgement for well-formed types and to include a level $n$, so that it has the form $\istypeU{\Gamma}{A}{n}$, and
likewise for type equality judgements. For example, 
\begin{mathpar}
  \inferrule*{
    \istypeU{\Gamma,\lock}{A}n}{
    \istypeU{\Gamma}{\Box A}n}
\end{mathpar}

In addition to the existing rules for types (indexed with a level) and terms, we then also include:
\begin{mathpar}
  \inferrule*{ \ }
            {\istypeU{\Emp}{\U_n}{n+1}}
  \and
  \inferrule*{\istypeU{\Gamma}{A}{n}}
            {\isterm{\Gamma}{\code{A}}{\U_n}}
  \and
  \inferrule*{\isterm{\Gamma}{u}{\U_n}}
            {\istypeU{\Gamma}{\E u}{n}}
\end{mathpar}
Finally, we add the following type and term equality rules: 
\begin{mathpar}
  \inferrule*{\istypeU{\Gamma}{A}{n}}
            {\eqtypeU{\Gamma}{\E\code{A}}{A}{n}}
  \and
  \inferrule*{\isterm{\Gamma}{u}{\U_n}}
            {\eqterm{\Gamma}{\code{\E u}}{u}{\U_n}}
\end{mathpar}

As an example, there is a term
\[
  \widehat\Box \defeq \lambda x.\code{\Box\E(\open\,x)} \;:\; \Box\U_n\to\U_n
\]
which encodes the $\Box$ type constructor on the universe in the sense that 
\[
\E(\widehat\Box(\shut\, u)) = \Box (\E u)
\]
This is similar to the $\widehat{\later}$ operator of Guarded Dependent
Type Theory~\cite{BirkedalL:gdtt-conf}, which is essential to defining
guarded recursive types. 
Thus, 
$\widehat{\later}$ arises for general reasons quite unconnected to the specifics
of guarded recursion.

\section{Discussion}
\label{sec:discussion}

\subsection{Related Work}

Modal dependent type theory builds on work on the computational
interpretation of modal logic with \emph{simple} types. Some of this work
involves a standard notion of context; most relevantly to this paper, the calculus
for Intuitionistic K of~\citeasnoun{Bellin:Extended}, which employs
\emph{explicit substitutions} in terms. Departing from standard contexts,
Fitch-style calculi were introduced independently by~\citeasnoun{Borghuis:Coming}
and~\citeasnoun{Martini:Computational}. Recent work
by~\citeasnoun{Clouston:fitch-2018} argued that Fitch-style calculus can be
extended to a variety of different modal logics, and gave a sound categorical
interpretation by modelling the modality as a right adjoint. Another non-standard
notion of context are the \emph{dual contexts} introduced by~\citeasnoun{Davies:Modal}
for the modal logic Intuituionistic S4 of comonads.
Here a context $\Delta;\Gamma$ is understood as meaning
$\Box\Delta\land\Gamma$, so the structure in the context is modelled by the
modality itself, not its left adjoint. Recent work by~\citeasnoun{kavvos2017dual} has
extended this approach to a variety of modal
logics, including Intuitionistic K.

There exists recent work employing variants of dual contexts for modal
\emph{dependent} type theory, all involving (co)monads rather than the more
basic logic of this paper. Spatial type theory~\cite{shulman2018brouwer}, designed for
applications in homotopy type theory (see
also~\cite{wellen2017formalizing,licata2018internal}), extends
the Davies-Pfenning calculus for a comonad with both dependent types and a
second modality, a monad right adjoint to the comonad. Second, the calculus
for parametricity of~\citeasnoun{Nuyts:Parametric} uses \emph{three}
zones to extend Davies-Pfenning with a monad \emph{left} adjoint to the
comonad. They focus on $\Pi$- and $\Sigma$-types with modalised arguments,
but a more standard modality can be extracted by taking the second argument
of a modalised $\Sigma$-type to be the unit type. In both the above works the
leftmost modality is intended to itself be a right adjoint, so they potentially could
also be captured by a Fitch-style calculus. Third,~\citeasnoun{dePaiva:Fibrational} suggest a generalisation of Davies-Pfenning with
some unusual properties, as $\Box$ types carry an auxiliary typed variable and
$\Pi$-types may only draw their argument from the modal context. We finally
note the dual contexts approach has inspired the mode theories
of~\citeasnoun{licata2017fibrational}, but this line of work as yet does not
support a term calculus.

We however do not know how to apply the dual context approach to modal logics
where the modality is not a (co)monad. For example it is not obvious how to extend
Kavvos's simply-typed calculus for Intuitionistic K. This should be compared to the ease of
extending the simply-typed Fitch-style calculus with dependent types. We hope
that Fitch-style calculi continue to provide a relatively simple setting for
modal dependent type theory as we explore the extensions discussed in the
next subsection.

We are not aware of any successful extensions of the explicit substitution
approach to dependent types; our own experiments with this while developing
Guarded Dependent Type Theory~\cite{clouston2015programming} suggests
this is probably possible but becomes unwieldy with real examples. Far more
succesful was the Clocked Type Theory~\cite{bahr2017clocks} discussed in
Section~\ref{sec:examples}, which can now be seen to have rediscovered the
Fitch-style framework, albeit with the innovation of named locks to control
fixed-point unfoldings. That work provides the inspiration for the more
foundational developments of this paper.


\subsection{Future work}

We wish to develop operational semantics for dependent Fitch-style calculi,
and conjecture that standard techniques for sound normalisation and
canonicity can be extended, as was possible for simply-typed Fitch-style
calculi~\cite{Borghuis:Coming,Clouston:fitch-2018}, and for Clocked Type
Theory~\cite{bahr2017clocks}. Such results should then lead to practical
implementation.

The modal axiom Intuitionistic K was used in this paper because it provides a
basic notion of modal necessity and holds of many useful models. Nonetheless
for particular applications we will want to develop Fitch-style calculi
corresponding to more particular logics. There can be no algorithm for
converting additional axioms to well-behaved calculi, but we know that
Fitch-style calculi are extremely versatile in the simply typed
case~\cite{Clouston:fitch-2018}, and Clocked Type Theory provides one
example of this with dependent types. In particular we are interested in
Fitch-style calculi with multiple interacting modalities, each of which is assigned
its own lock; we hope to develop guarded type theory with both $\later$ and
$\Box$ modalities in this style.

The notion of CwF with a weak CwF endomorphism
(Definition~\ref{def:Gir-map}) is more general than our CwF+A, as it does
not require the existence of a left adjoint. Because a  weak CwF endomorphism must
preserves products, it appears to be a rival candidate for a model of dependent type theory
with the K axiom. However we do not know how to capture this class of models in syntax.
Understanding this would be valuable because
\emph{truncation}~\cite{DBLP:journals/logcom/AwodeyB04},
considered as an endofunctor for example on sets, defines such a morphism but
is not a right adjoint. Truncation allows one to move between general types and 
propositions. For example combining it with guarded types would allow us to
formalise work in this field that makes that
distinction~\cite{birkedal2011first,clouston2015programming}.

\bibliographystyle{agsm}
\bibliography{drat}

\newpage

\end{document}